\documentclass[letterpaper, 10 pt, conference]{ieeeconf}  

\IEEEoverridecommandlockouts                              

\usepackage[english]{babel}
\usepackage{algorithm}
\usepackage{algorithmic}
\usepackage{amsmath}
\usepackage{amsthm}
\usepackage{amssymb}
\usepackage{mathtools}
\usepackage{changes}
\usepackage{times} 
\usepackage{helvet}  
\usepackage{courier}  
\usepackage[hyphens]{url}  
\usepackage{graphicx} 
\usepackage{caption} 
\usepackage{acronym}
\usepackage{xcolor}
\usepackage{paralist}
\usepackage{changes}
\usepackage{pdfpages}
\usepackage{caption}
\usepackage{subcaption}
\usepackage[nocompress]{cite}

\title{\LARGE \bf
Active  Perception with Initial-State Uncertainty: A Policy Gradient Method
}

\author{Chongyang Shi$^{1}$, Shuo Han$^{2}$, Michael Dorothy$^{3}$, and Jie Fu$^{1}$
\thanks{$^{1}$Chongyang Shi and Jie Fu are with the Department of Electrical and Computer Engineering, University of Florida, Gainesville, FL, USA.
        {\tt\small \{c.shi, fujie\}@ufl.edu }}%
\thanks{$^{2}$Shuo Han is with the Department of Electrical and Computer Engineering, University of Illinois Chicago, Chicago, IL, USA.
        {\tt\small hanshuo@uic.edu}}%
\thanks{$^{3}$Michael Dorothy is with DEVCOM Army Research Laboratory, Adelphi, MD, USA.
        {\tt\small michael.r.dorothy.civ@army.mil}}%
        \thanks{Research was sponsored by the Army Research Laboratory  under  W911NF-22-2-0233 and in part by the Army Research Office under W911NF-22-1-0166.}
}

  \newcommand{\nat}{\mathbf{N}}

  \newcommand{\ie}{\textit{i.e.}}
  \newcommand{\eg}{\textit{e.g.}}

\newcommand{\dist}{\mathcal{D}}

\newcommand{\reals}{\mathbb{R}}
\newcommand{\probs}{\mathbb{P}}

\newtheorem{theorem}{Theorem} 
\newtheorem{definition}{Definition}
\newtheorem{proposition}{Proposition}

\newtheorem{problem}{Problem}

\newtheorem{remark}{Remark}

\newtheorem{assumption}{Assumption}

\acrodef{mdp}[MDP]{Markov decision process} 

 \DeclareMathOperator*{\optmin}{\mathrm{minimize}}

  \newcommand{\norm}[1]{\lVert #1 \rVert} 
 
\newcommand{\Expect}{\mathbb{E}}

\newcommand{\calS}{\mathcal{S}}
\newcommand{\calA}{\mathcal{A}}

\newcommand{\calO}{\mathcal{O}}

\newcommand{\calY}{\mathcal{Y}}

\acrodef{hmm}[HMM]{hidden Markov model}
\acrodef{fsc}[FSC]{finite state controller}
\acrodef{pomdp}[POMDP]{partially observable Markov decision process} 

\newcommand{\matO}{{\mathbf{O}}}
\newcommand{\matT}{{\mathbf{T}}}
\newcommand{\matA}{{\mathbf{A}}}

\newcommand{\deq}{\triangleq}

\definecolor{darkgreen}{rgb}{0,0.5,0}

\begin{document}

\maketitle

\begin{abstract}

This paper studies the synthesis of an active perception policy that maximizes the information leakage of the initial state in a stochastic system modeled as a hidden Markov model (HMM). Specifically, the emission function of the HMM is controllable with a set of perception or sensor query actions. Given the goal is to infer the initial state from partial observations in the HMM, we use Shannon conditional entropy as the planning objective and develop a novel policy gradient method with convergence guarantees. By leveraging a variant of observable operators in HMMs, we prove several important properties of the gradient of the conditional entropy with respect to the policy parameters, which allow efficient computation of the policy gradient and stable and fast convergence. We demonstrate the effectiveness of our solution by applying it to an inference problem in a stochastic grid world environment.
 
\end{abstract}

%
\section{INTRODUCTION}


   This paper studies the synthesis of an active perception strategy that maximizes the transparency of an initial state in a stochastic system given partial observations.  We introduce a set of active perception actions into the system modeled as a \ac{hmm} such that the emission at a given state is jointly determined by the state and a perception action. The goal is to compute an active perception strategy that maximizes the information leakage about the initial state $S_0$ given observations $Y$, which is measured by the negative conditional entropy $H(S_0|Y)$. 

 The contributions of this work are summarized as follows: First, we prove that active perception planning to minimize the conditional entropy $H(S_0|Y)$ cannot be reduced to a partially observable Markov decision process (POMDP) with a belief-based reward function. Leveraging a variant of observable operators \cite{jaegerObservableOperatorProcesses1997},  we develop an efficient algorithm to compute the posterior distribution $\probs(S_0|Y)$ given a perception policy. Additionally, it is shown that the policy gradient of conditional entropy depends only on the posterior distribution $\probs(S_0|Y)$ and the gradient of the policy with respect to its parameters. We prove that the entropy is Lipschitz continuous and Lipschitz smooth in the policy parameters under some assumptions for policy search space and thus ensure the convergence of the gradient-based planning. Finally, we evaluate the performance in a stochastic grid world environment.



\noindent \textbf{Related Work} 
Active perception \cite{bajcsyRevisitingActivePerception2018} is to gather information selectively for improving the task performance of an autonomous system. The applications of active perception range from object localization \cite{andreopoulos2009theory}, target tracking \cite{CASAO2024103876,zhou2011multirobot} to mission planning for surveillance and monitoring \cite{lauriMultiagentActivePerception2020,6161683}. 
Information-theoretic metrics are introduced as planning objectives in various active perception problems. 
 One approach~\cite{egorovTargetSurveillanceAdversarial2016,Araya2010Neurips} to active perception is to formulate a \ac{pomdp} with a reward function that depends on the information state or the belief of an agent. For example, in target surveillance, a patrolling team is rewarded by uncertainty reduction of the belief about the intruder's position/state. In intent inference, the authors \cite{Shendeep2019} used the negative entropy of the belief over the opponent's intent as a reward and maximized the total reward for active perception.

Besides active perception, several studies have explored decision-making with information-theoretic objectives.
In \cite{molloy2023smoother}, the authors introduce a method for obfuscating/estimating state trajectories in \ac{pomdp}s and show that the causal conditional entropy of the state trajectory given observations and controls can be reformulated as a cumulative sum, allowing the use of standard \ac{pomdp} solvers. The work \cite{shi2024information} develops a method to 
maximize the conditional entropy of a secret variable in an MDP to a passive observer. Another work \cite{SavasEntropy2022} proposes entropy maximization in POMDPs to minimize the predictability of an
agent’s trajectories to an outside observer. In both cases, the planning agent controls the stochastic dynamics and the observer is passive. In comparison, this paper studies the dual problem when the control system is autonomous but the observer is active. Thus, the observer's policy is restricted to be observation-based. Unlike \cite{SavasEntropy2022} where the goal is to maximize the sum of the conditional entropy of current states given the historical states, our goal is to maximize the information about some past state (initial state) given the observations received by the observing agent. This problem has important applications in intent recognition and system diagnosis \cite{lafortuneHistoryDiagnosabilityOpacity2018a}. 

\section{PRELIMINARIES AND PROBLEM FORMULATION}
\label{sec:prelim}
\noindent \textbf{Notation} The set of real numbers is denoted by $\reals$. Random variables will be denoted by capital letters, and their realizations by lowercase letters (\eg, $X$ and $x$).  A sequence of random variables and their realizations with length $T$ are denoted as $X_{0:T}$ and $x_{0:T}$. The notation $x_i$ refers to the $i$-th component of a vector $x \in \reals^{n}$ or to the $i$-th element of a sequence $x_0, x_1, \ldots$, which will be clarified by the context. 
Given a finite set $\mathcal{S}$, let $\dist(\mathcal{S})$ be the set of all probability distributions over $\mathcal{S}$. The set $\mathcal{S}^{T}$ denotes the set of sequences with length $T$ composed of elements from $\mathcal{S}$, and $\mathcal{S}^\ast$ denotes the set of all finite sequences generated from $\mathcal{S}$.   The empty string in $\mathcal{S}^\ast$ is denoted by $\lambda$.

We introduce a class of active perception problems where the agent cannot control the dynamical system but can select perception actions to monitor it in order to infer some unknown state variables. The class of perception actions can be the choices of sensors to query in a distributed sensor network or the choices of poses for cameras with a limited field of view (FoV).
\begin{definition}
An \ac{hmm} with a controllable emission function is a tuple $
\mathcal{M} = \langle \calS , \calO, \calA, P, E, S_0 \rangle$
where 
\begin{inparaenum}
    \item $\calS=\{1,\ldots, N\}$ is a finite  state space. 
    \item $\calO$ is a finite set of observations.
    \item $\calA$ is a finite set of \emph{perception actions}. 
    \item $P:\calS\rightarrow \dist(\calS)$ is the probabilistic transition function.
    \item $E: \calS\times \calA \rightarrow \dist(\calO)$ is the emission function (observation function) that takes a state $s$ and a \emph{perception action} $a$, outputs a distribution over observations.  
    \item $S_0$ is a random variable representing the initial state. The distribution of $S_0$ is denoted by $\mu_0$. And $\mathcal{S}_0$ denotes the set of initial states.
\end{inparaenum}
\end{definition} 
 
\begin{definition} 
\label{def:perception_policy_general}
A \emph{non-stationary, observation-based} perception policy is a function $\pi: \calO^\ast \to \dist(\calA)$ that maps a history $o_{0:t}$ of observations to a distribution $\pi(o_{0:t})$ over perception actions.
\end{definition}
\begin{definition}
\label{def:perception_policy}
    A \emph{finite-state, observation-based} perception policy $\pi$ with deterministic transitions is a tuple $\pi\coloneqq \langle  Q, \calO, \calA, \delta, \psi, q_0 \rangle $ where \begin{inparaenum}
\item $Q$ is a set of memory states.
\item $\calO, \calA$ are a set of inputs and a set of outputs, respectively.
\item $\delta: Q\times \calO \rightarrow Q$ is a deterministic transition function that maps a state and an input $(q,o)$ to a next state $\delta(q,o)$. 
\item $\psi:Q\rightarrow \dist(\calA) $ is a probabilistic output function.
\item $q_0$ is the initial state.
\end{inparaenum}
\end{definition}
 Given a sequence $o_{0:t}$ of observations, the non-stationary policy directly outputs a distribution   $\pi(o_{0:t})$. A finite-state policy first computes the reached memory state $q= \delta(q_0, o_{0:t})$ and then outputs a distribution $\psi(q)$ over actions. In both cases, we can treat $\pi$ as a function that maps $o_{0:t}$ to a distribution over actions and write $\pi(a| o_{0:t})$ as the probability of taking action $a$ given observations $o_{0:t}$. 
When there is no observation, \ie, before the initial observation is made, the agent selects a perception action according to $\pi(\lambda)$. If $\pi$ is a finite-state policy, then we can further obtain $\pi(\lambda) = \pi(\psi(q_0))$ because $\delta(q_0,\lambda)= q_0$.

For a given \ac{hmm} $\mathcal{M}$, a perception policy $\pi$ induces a discrete stochastic process $\{S_t, A_t,  O_t, t\in \nat\}$, where $S_t \in \mathcal{S}$ and $O_t \in \mathcal{O}$ are the underlying hidden state and observation at the $t$-th time step, and   
\begin{multline*}
    \probs^\pi(O_t = o | O_{0:t-1}= o_{0:t-1}, S_t=s_t) =\\  \sum_{a\in \calA} E(o| s_t, a )\pi(a| o_{0:t-1}), \forall t>0.
\end{multline*}
when the perception policy $\pi$ is understood from the context, we write $\probs$ instead of $\probs^\pi$ for clarity.

The conditional entropy of $X_2$ given $X_1$ is defined by
\[
H(X_2|X_1) =   -\sum_{x_1\in \mathcal{X}}\sum_{x_2\in \mathcal{X}} p(x_1,x_2) \log p(x_2|x_1).
\]
The conditional entropy measures the uncertainty about $X_2$ given knowledge of $X_1$.  
A lower conditional entropy makes it easier to learn $X_2$ from observing a sample of $X_1$. 

For any finite horizon $T$, the agent's partial observation about a path in the \ac{hmm} includes a sequence $o_{0:T}$ of observations and a sequence $a_{0:T}$ of perception actions. We denote the agent's information by $y_{0:T}=(o_{0:T}, a_{0:T})$. When the length of $y_{0:T}$ is clear from the context, we omit the subscript and use $y$ to denote the sequence. In the following, we refer to $y$ as an observation sequence with the understanding that it is the joint observation and perception action sequence.

\begin{problem} 
\label{problem:initial-state}
    Let an \ac{hmm} $\mathcal{M}$ and a finite horizon $T$ be given. Let $\Pi$ be a policy space.
    Compute an active perception policy $\pi \in \Pi$ that minimizes the conditional entropy of the initial state $S_0$ given the partial information $Y_{0:T}= (O_{0:T}, A_{0:T})$ induced by $\pi$. That is,
    \begin{multline*}
 \optmin_{\pi \in \Pi} \quad       H(S_0  |Y_{0:T}; {M_\pi}) \triangleq \\
 - \sum_{s_0 \in \mathcal{S}_0, y_{0:T} \in \calO^T\times \calA^T}  \probs^\pi (s_0, y_{0:T} )\log \probs^\pi (s_0\mid y_{0:T}),
    \end{multline*}   
where $\probs^\pi (s_0, y_{0:T} )$  is the joint probability of starting from $s_0$ and observing $y_{0:T}$ under the policy $\pi$ and $\probs^\pi (s_0\mid y_{0:T})$ is the conditional probability of starting from $s_0$ given observation $y_{0:T}$.
\end{problem}

\begin{remark}
A constrained formulation to minimize entropy given a bounded perception cost can also be formulated. Because policy gradient methods \cite{agarwal2021theory} for optimizing a cumulative reward/cost are well-understood, we focus on solving this entropy minimization problem and expect that only a small modification to the gradient computation is needed to minimize a weighted sum of the conditional entropy and the expected total cost of perception actions. 
\end{remark}

 
\section{MAIN RESULT}
 First, we show that the problem cannot be reduced to a $\rho$-POMDP \cite{Araya2010Neurips}  which is a POMDP with a belief-based reward. Given observations $y_{0:t} = (o_{0:t}, a_{0:t})$ and a perception policy $\pi$, the belief about $S_0$ is the posterior distribution $b_t(s) = \probs^{\pi}(S_0=s_0|o_{0:t}, a_{0:t} )$.  
\begin{proposition}
\label{prop:no-belief-reward}
 There is no belief-based reward $R:\Delta(S)\rightarrow \reals$ such that $\sum_{i=0}^t R(b_i) = - H(S_0|y_{0:t}) $.
\end{proposition}
\begin{proof}
 Suppose, by way of contradiction, such a belief-based reward function exists. Then, $R(b_{t }) = \sum_{i=0}^t R(b_i) - \sum_{i=0}^{t-1} R(b_i) =- H(S_0|y_{0:t}) +   H(S_0|y_{0:t-1}) $. We show that it is possible to reach the same belief with two different observations $y_{0:t}$ and $y'_{0:t'}$, but $- H(S_0|y_{0:t} ) +   H(S_0|y_{0:t-1}) \ne - H(S_0|y'_{0:t'} ) +   H(S_0|y'_{0:t'-1})$. Consider an example of HMM where the initial state can be either $0$ or $1$ with a prior distribution $[0.5, 0.5]$; namely, the initial belief is given by $b_0 = [0.5, 0.5]$. If $S_0=0$, the next state $S_1$ will always emit an observation $0$. If $S_0=1$, the next state $S_1$ will always emit an observation $1$. Thus, if $0$ is observed, then $b_1= [1,0]$; if $1$ is observed, then $b_1= [0,1]$. The reward $R(b_1) = - H(S_0|o) +H(S_0|\lambda) = 1$ regardless whether $o=0$ or $o=1$ is observed. After reaching state $S_1$, the system reaches $S_2$ next and yields some observation $o'$. The belief does not change, \ie, $b_2=b_1$ because the initial state is known with certainty. However, based on the formula $R(b_2) =   - H(S_0|o' o) +H(S_0|o ) = -0+0=0$, which is different from $R(b_1)$. A contradiction is established.
\end{proof}




In the next, we develop a policy gradient method to solve Problem~\ref{problem:initial-state}. Consider a class of parametrized (stochastic) policies $\{\pi_\theta\mid \theta\in \Theta\}$. We denote by $M_\theta$ the stochastic process $\{S_t,A_t,O_t, t\in \nat\}$ induced by a policy $\pi_\theta$, and $\probs_\theta(\cdot)$ the corresponding probability measure.


For a given parameterized policy $\pi_\theta$, we denote by $\calY_\theta = \{y\in O^T\times A^T| P_\theta (y) >0\}$ the set of possible observations under the perception policy $\pi_\theta$.  
The gradient of $H(S_0|Y; \theta) \deq H(S_0|Y; \pi_\theta)$ is given by
\begin{equation}
\begin{aligned}
\label{eq:HMM_gradient_entropy}
 &\nabla_\theta H(S_0|Y;\theta) \\
= & - \sum_{y \in \calY_\theta } \sum_{s_0 \in \mathcal{S}_0} \Big[\nabla_\theta \probs_\theta(s_0, y) \log \probs_\theta(s_0 | y) \\
&+  \probs_\theta(s_0, y) \nabla_\theta  \log \probs_\theta(s_0 | y)\Big] \\
= & - \sum_{y \in \calY_\theta} \sum_{s_0 \in \mathcal{S}_0} \Big[\probs_\theta(y) \nabla_\theta \probs_\theta(s_0| y) \log \probs_\theta(s_0 | y) + \\ 
& \probs_\theta(s_0| y) \nabla_\theta \probs_\theta(y) \log \probs_\theta(s_0 | y) 
+   \probs_\theta(y)\nabla_\theta \probs_\theta(s_0 | y)\Big]\\
 = & - \sum_{y \in \calY_\theta} \probs_\theta (y) \sum_{s_0 \in \mathcal{S}_0} \Big[ \log \probs_\theta(s_0 | y) \nabla_\theta \probs_\theta(s_0| y) \\
& + \probs_\theta(s_0| y) \log \probs_\theta(s_0 | y) \nabla_\theta \log \probs_\theta(y) + \nabla_\theta \probs_\theta(s_0 | y) \Big]. 
\end{aligned}
\end{equation}
In the following, Propositions~\ref{prop:grad_is_zero} and \ref{prop:two_grads_equal} will allow us to further simplify the computation of gradient.


To derive the results, we introduce the \emph{observable operator} augmented with perception actions. The observable operator \cite{jaeger2000observableoperator} has been proposed to represent a discrete \ac{hmm} and used to calculate the probability of an observation sequence in an \ac{hmm} using matrix multiplications.


Let the random variable of state, observation, and action, at time point $t$ be denoted as $X_t, O_t, A_t$, respectively. Let $\matT \in \reals^{N \times N}$ be the flipped state transition   matrix with $$\matT_{i,j} = \probs(X_{t+1} = i|X_t = j).$$ 
For each $a\in \calA$, Let $\matO^a \in \reals^{M \times N}$ be the observation probability matrix with $\matO_{o,j}^{a} =E( o |  j, a)$.
\begin{definition} 
Given the \ac{hmm} with controllable emissions $\mathcal{M}$, for any pair of observation and perception action $(o,a)$,
 the observable operator given perception actions $\matA_{o|a}$ is a matrix of size $N \times N$ with its $ij$-th entry defined as $$
 \matA_{o |a }[i,j] =  \matT_{i , j}\matO_{o,j}^{a} \ ,$$
 which is the probability of transitioning from state $j$ to state $i$ and at the state $j$,  an observation $o$ is emitted given   perception action $a$.
  In matrix form, 
\[
\matA_{o |a } = \matT \text{diag}(\matO_{o, 1}^{a }, \dots, \matO_{o , N}^{a }).
\]
\end{definition}



 
\begin{proposition}
\label{prop:observation-action-probability} The probability of an observation sequence $o_{0:t}$ given a sequence of perception actions $a_{0:t}$, can be written as matrix operations,
\begin{equation}
\label{eq:matrix_operation}
\probs(o_{0:t} | a_{0:t}) = \mathbf{1}_N^\top \matA_{o_t|a_t} \dots \matA_{o_0|a_0} \mu_0.
\end{equation}
In addition, 
for a fixed initial state $s_0\in \calS_0$,

\begin{equation}
\label{eq:matrix_operation_s0}
\probs(o_{0:t} | a_{0:t},s_0) = \mathbf{1}_N^\top \matA_{o_t|a_t} \dots \matA_{o_0|a_0} \mathbf{1}_{s_0}.
\end{equation}
where $\mathbf{1}_{s_0}$ is a one-hot vector which assigns 1 to the $s_0$-th entry.
\end{proposition}

\begin{proof}Since the transition does not depend on the perception actions, the following equation holds 
according to the matrix notation of the well-known ``forward algorithm" in \cite{jaeger2000observableoperator}, 
\begin{equation*}
\begin{aligned}
\probs(o_{0:t} | a_{0:t}) &= \mathbf{1}_N^\top \matT \text{diag}(O^{a_t}_{o_t, 1}, \dots, O^{a_t}_{o_t, N}) \cdot \dots \\
& \cdot \matT \text{diag}(O^{a_0}_{o_0, 1}, \dots, O^{a_0}_{o_0, N}) \mu_0.
\end{aligned}
\end{equation*}
\eqref{eq:matrix_operation_s0}
is derived by replacing the initial distribution $\mu_0$ with the one-hot distribution. 
\end{proof}
 
To compute the gradient $\nabla_\theta H(S_0|Y;\theta)$, we will need the value of $\nabla_\theta \log \probs_\theta(y)$. We start by calculating the probability of an observation sequence $y=(o_{0:t}, a_{0:t}) $ in $M_\theta$. 
\begin{proposition}
The probability of a sequence  $y=(o_{0:t}, a_{0:t}) $ of observations and perception actions in $M_\theta$ can be computed as follows:
\begin{equation}
\begin{aligned}
\label{eq:joint_final}
\probs_\theta(o_{0:t}, a_{0:t}) & =  \frac{\probs(o_{0:t} | a_{0:t}) }{\probs(o_{0} | a_{0}) }  \prod_{i=0}^{t} \pi_\theta(a_{i}| o_{0:i-1}). 
\end{aligned}
\end{equation}
where $o_{0:-1} \coloneqq \lambda$ is the initial empty  observation.
\end{proposition}
\begin{proof}
By the product rule of probability and causality, we can write the probability in the form,
 \begin{align}
& \probs_\theta(o_{0:t}, a_{0:t})  \nonumber \\
& =   \probs_\theta(o_{t}, a_{t} | o_{0:t-1}, a_{0:t-1}) \nonumber \\
& \qquad \cdot \probs_\theta(o_{t-1}, a_{t-1} | o_{0:t-2}, a_{0:t-2}) \cdot \ldots  \cdot \probs_\theta(o_{1}, a_{1} | o_0, a_0) \nonumber \\
& =  \prod_{i=1}^{t} \probs_\theta(o_i, a_i| o_{0:i-1},a_{0:i-1}).\label{eq:joint_distribution}
\end{align}
 For any $1 \le i \le t, \; i \in \mathbb{N}$, based on the multiplication rule of probability, we can decompose the conditional probability $\probs(o_{i}, a_{i} | o_{0:i-1}, a_{0:i-1})$ as
\begin{equation}
\label{eq:decompos}
\probs_\theta(o_{i}, a_{i} | o_{0:i-1}, a_{0:i-1}) = \probs (o_i | a_{0:i}, o_{0:i-1}) \pi_\theta(a_{i}| o_{0:i-1}). 
\end{equation}
Note that the probability of observing $o_i$ given the sequence of actions $a_{0:i}$ and past observations $o_{0:i-1}$ is independent from the policy $\pi_\theta$. 
And the conditional probability $\probs(o_i | a_{0:i}, o_{0:i-1})$ is derived as 
\begin{align}
      & \probs(o_i\mid a_{0:i},o_{0:i-1}) \nonumber \\
   = &  \sum_{s_i \in \calS} \probs(o_i|s_i,a_{0:i},o_{0:i-1}) 
 \probs(s_i|a_{0:i}, o_{0:i-1})  \nonumber \\
 \overset{(i)}{=}  &  \sum_{s_i\in \calS} E(o_i|s_i,a_i) 
 \probs(s_i|a_{0:i-1}, o_{0:i-1})  \nonumber \\
     = & \sum_{s_i\in \calS} E(o_i|s_i,a_i)\frac{\probs(s_i, o_{0:i-1}|a_{0:i-1})}{\probs(o_{0:i-1}|a_{0:i-1})} \nonumber \\
    = & \frac{1}{\probs(o_{0:i-1}|a_{0:i-1})}  \sum_{s_i\in \calS} E(o_i|s_i,a_i) \probs(s_i, o_{0:i-1}|a_{0:i-1}) \nonumber \\
    \overset{(ii)}{=} & \frac{\probs(o_{0:i}|a_{0:i})}{\probs(o_{0:i-1}|a_{0:i-1})}. \label{eq:prob-observation-i-cond} 
\end{align}
where (i) is because 1) the probability of observing $o_i$ is determined by state $s_i$ and perception action $a_i$ given the emission function $E(\cdot)$; and 2) the probability of reaching state $s_i$ at the $i$-th time step   does not depend on the perception action $a_i$ when the action sequence $a_{0:i}$ is fixed,
which can be calculated by equation~\eqref{eq:matrix_operation}. The equality $(ii)$ is established by the definition of observable operators and Proposition~\ref{prop:observation-action-probability}. 
Substituting \eqref{eq:prob-observation-i-cond} into \eqref{eq:decompos} and rewrite \eqref{eq:joint_distribution}, 
\begin{equation}
\begin{aligned}
     \probs_\theta(o_{0:t}, a_{0:t}) &= \prod_{i=1}^{t} \frac{\probs(o_{0:i} | a_{0:i}) }{\probs(o_{0:i-1} | a_{0:i-1}) } \pi_\theta(a_{i}| o_{0:i-1}) \\
     &=  \frac{\probs(o_{0:t} | a_{0:t}) }{\probs(o_{0} | a_{0}) }  \prod_{i=0}^{t} \pi_\theta(a_{i}| o_{0:i-1}).
\end{aligned}
\end{equation}
which can be computed efficiently using the matrix operators for calculating  $\probs(o_{0:t} | a_{0:t}) $.  
\end{proof}

For a fixed initial state $s_0 \in \calS_0$, the result of~\eqref{eq:joint_final}  becomes
\begin{equation}
\label{eq:joint_calcu}
\probs_\theta(y|s_0) = \frac{\probs(o_{0:t} | a_{0:t}, s_0) }{\probs(o_{0} | a_{0}, s_0) } \prod_{i=0}^{t} \pi_\theta(a_{i}| o_{0:i-1}),
\end{equation}
where the calculation of term $\probs(o_{0:t} | a_{0:t}, s_0) $ is given in~\eqref{eq:matrix_operation_s0}.

Numerical issues may arise in computing $\probs_\theta(y|s_0)$ and $\probs_\theta(y)$ because the probabilities can be close to 0 for a relatively long horizon $t$. We can avoid these numerical issues by taking the logarithm of both sides of equation \eqref{eq:joint_calcu}. 
The following properties further show the gradient calculation can be simplified.
\begin{proposition}
\label{prop:two_grads_equal}
Given $y= (o_{0:T}, a_{0:T})$, 
the gradient of $\log P_\theta(y|s_0)$ can be computed as
\begin{equation}
\label{eq:two_grads_equals}
\begin{aligned}
\nabla_\theta \log  \probs_\theta(y |s_0)& = \nabla_\theta \log  \probs_\theta(y )\\
&=  \sum_{t = 0}^{T} 
\nabla_\theta \log  \pi_\theta(a_{t}| o_{0:t-1}).
\end{aligned}
\end{equation}
which is invariant with respect to the initial state $s_0$.
\end{proposition}
\begin{proof}
First, let us take the logarithm on both sides in \eqref{eq:joint_calcu}, 
\begin{equation*}
\begin{aligned}
\log  \probs_\theta(y|s_0)& = \log  \probs(o_{0:t}| a_{0:t}, s_0) - \log  \probs(o_{0}| a_{0}, s_0)\\
&+ \sum_{t = 0}^T \log  \pi_\theta(a_{t}| o_{0:t-1}).
\end{aligned}
\end{equation*}
Then, we calculate the gradient on both sides,
\begin{equation}
\label{eq:nabla_log_p_theta_y_g_s0}
\nabla_\theta \log  \probs_\theta(y|s_0) = \sum_{t = 0}^T
\nabla_\theta \log \pi_\theta(a_{t}| o_{0:t-1}).
\end{equation}
which can be shown to equal $\nabla_\theta \log \probs_\theta(y)$ using  \eqref{eq:joint_final}.
\end{proof}

\begin{proposition}
\label{prop:grad_is_zero}
For any $y \in \calY_\theta$,
the gradient of the logarithm of the posterior probability $ \probs_\theta(s_0|y )$ with respect to $\theta$ is 0, \ie, 
$
\nabla_\theta \log  \probs_\theta(s_0|y ) = 0.$
Further, when $\probs_\theta(s_0|y ) \ne 0$, 
 $
\nabla_\theta \probs_\theta(s_0|y ) = 0.$ 
\end{proposition} 
\begin{proof}
First, using the Bayes' rule,
\[\probs_\theta(s_0|y) =   \frac{\probs_\theta(y|s_0)\mu_0(s_0)}{\probs_\theta(y)} \]
Taking the logarithm on both sides:
\begin{multline*}
\log \probs_\theta(s_0|y) = \log \probs_\theta(y|s_0)  
+\log \mu_0(s_0) - \log \probs_\theta(y),
\end{multline*}
and then taking the gradient on both sides with respect to $\theta$, 
\begin{multline*}
\nabla_\theta \log \probs_\theta(s_0|y) =\nabla_\theta \log \probs_\theta(y|s_0) \\+\nabla_\theta \log \mu_0(s_0) - \nabla_\theta \log \probs_\theta(y),
\end{multline*}
From Proposition~\ref{prop:two_grads_equal}, we derive $\nabla_\theta \log \probs_\theta(s_0|y) =0$ because $\nabla_\theta \log \probs_\theta(y|s_0)  - \nabla_\theta \log \probs_\theta(y) =0$ and $\log \mu_0(s_0)$ is a constant.  
Furthermore, because $\nabla_\theta \log \probs_\theta(s_0|y) = \frac{\nabla_\theta \probs_\theta(s_0|y)}{\probs_\theta(s_0|y)}$, when $\probs_\theta(s_0|y)\ne 0$ but $\nabla_\theta \log \probs_\theta(s_0|y)=0$, we can derive that $\nabla_\theta \probs_\theta(s_0|y)=0$.
\end{proof}

\begin{theorem}
\label{thm:entropy-grad}
The gradient of the conditional entropy w.r.t. the policy parameter $\theta$ is 
\begin{equation}
\label{eq:entropy-grad-expectation}
\begin{aligned}
&\nabla_\theta H (S_0|Y;\theta) =  \Expect_{y\sim M_\theta} \left[ H(S_0|Y= y;\theta)\nabla_\theta \log \probs_\theta(y )\right].
\end{aligned}
\end{equation}
\end{theorem}
\begin{proof}
The conditional entropy is given by:
\[
H(S_0 | Y; \theta) = -\sum_{y \in \mathcal{O}^T} \probs_\theta(y) \sum_{s_0 \in \mathcal{S}_0} \probs_\theta(s_0 | y) \log \probs_\theta(s_0 | y).
\]
When computing the gradient using \eqref{eq:HMM_gradient_entropy}, the terms in the summation corresponding to $ \probs_\theta(s_0 | y) = 0 $ vanish. Let $\mathcal{S}_{0,y} = \{s_0 \in \mathcal{S} \mid \probs_\theta(s_0| y) > 0\}$ be the set of initial states from which an observation $y$ is possible. Then,
\begin{equation*}
\begin{aligned}
 &\nabla_\theta H(S_0|Y;\theta) \\
 = & - \sum_{y \in \calY_\theta} \probs_\theta (y) \sum_{s_0 \in \mathcal{S}_{0,y}} \Big[ \log \probs_\theta(s_0 | y) \nabla_\theta \probs_\theta(s_0| y) \\
& + \probs_\theta(s_0| y) \log \probs_\theta(s_0 | y) \nabla_\theta \log \probs_\theta(y) + \nabla_\theta \probs_\theta(s_0 | y) \Big] 
\end{aligned}
\end{equation*}
From Proposition~\ref{prop:grad_is_zero}, the gradient $\nabla_\theta \probs_\theta(s_0| y) = 0$ when $ \probs_\theta(s_0| y) \neq 0$. Thus, we have
\begin{equation}
\label{eq:entropy-grad}
\begin{aligned}
& \nabla_\theta H (S_0|Y;\theta) =  - \sum_{y \in \calY_\theta} \probs_\theta (y) \\
& \sum_{s_0 \in \mathcal{S}_{0,y}} \Big[ \probs_\theta(s_0| y ) \log \probs_\theta(s_0 | y ) \nabla_\theta \log \probs_\theta(y ) \Big],\\
& = \Expect_{y\sim M_\theta} \left[ H(S_0|Y= y;\theta)\nabla_\theta \log \probs_\theta(y )\right]
\end{aligned}
\end{equation}
where the expectation is taken with respect to the stochastic process of observations induced by the perception policy $\pi_\theta$. Note that $\nabla_\theta \log \probs_\theta(y )$ can be computed using \eqref{eq:two_grads_equals}. 
\end{proof}    


Next, we show the convergence of a gradient-descent method under a common assumption of the policy space.

\begin{assumption}
\label{assume:bounded-policy-gradient}
For any time step $t \in [0,T]$, for any $(o_{0:t-1},a)\in O^{t}\times A$,   both $\| \nabla_\theta \log \pi_\theta(a | o_{0:t-1}) \|$ and $\|\nabla_\theta^2 \log \pi_\theta(a_{t}| o_{0:t-1}) \|$ are bounded.
\end{assumption}
\begin{theorem}
Under Assumption~\ref{assume:bounded-policy-gradient}, 
the   entropy $H(S_0|Y; \theta)$ is Lipschitz-continuous and Lipschitz-smooth in $\theta$.
\end{theorem}
\begin{proof}
We prove the theorem by showing that the gradient $\nabla_\theta H(S_0|Y; \theta)$ and Hessian $\nabla_\theta^2 H(S_0|Y; \theta)$ are both bounded. 
Referring to \eqref{eq:entropy-grad},  
by Jensen's inequality, we obtain
\begin{equation}
\label{eq:jensen}
H(S_0| Y = y; \theta) \leq \log |\mathcal{S}_0|. 
\end{equation}
The results of Proposition~\ref{prop:two_grads_equal} show that
\begin{equation}
\label{eq:bounded-gradient}
\nabla_\theta \log  \probs_\theta(y) = \sum_{t = 0}^T
\nabla_\theta \log \pi_\theta(a_{t}| o_{0:t-1}).
\end{equation}
which is bounded given Assumption~\ref{assume:bounded-policy-gradient} and the triangle inequality.
Due to the boundedness of $H(S_0| Y = y; \theta)$ and $ \nabla_\theta \log \probs_\theta(y)$, the gradient $ \nabla_\theta H(S_0| Y = y; \theta)$ is also bounded because $\Expect_{y\sim M_\theta} \left[ H(S_0|Y= y;\theta)\nabla_\theta \log \probs_\theta(y )\right] \le \Expect_{y\sim M_\theta} ( \beta  ) =  \beta $ where $\beta = \max_{y \in \calY_\theta} \norm{ H(S_0|Y= y;\theta)\nabla \log \probs_\theta(y)}_\infty$.
Next, consider the Hessian
\begin{equation}
\begin{aligned}
\nabla_\theta^2 H(S_0|Y; \theta) = &\Expect_{y\sim M_\theta} \left[ \nabla_\theta H(S_0|Y= y;\theta) \nabla_\theta \log \probs_\theta(y) \right.\\
 & \left.+  H(S_0|Y= y;\theta) \nabla_\theta^2 \log \probs_\theta(y)\right]\\
= &\Expect_{y\sim M_\theta}\left[H(S_0|Y= y;\theta) \nabla_\theta^2 \log \probs_\theta(y)\right]
\end{aligned}
\end{equation}
where the last equality is because 
  $\nabla_\theta \probs_\theta(s_0|y) = \nabla_\theta  \log \probs_\theta(s_0|y) = 0$ and thus  $\nabla_\theta H(S_0|Y= y;\theta)  =0 $ by Proposition~\ref{prop:grad_is_zero}. Further, when   $\|\nabla_\theta^2 \log \pi_\theta(a_{t}| o_{0:t-1}) \|$ is bounded for all $t \in [0,T]$, $o_{0:t-1}\in O^t$, and $a\in A$,  
\begin{equation}
\label{eq:bounded-hessian}
\nabla_\theta^2 \log  \probs_\theta(y) = \sum_{t = 0}^T
\nabla_\theta^2 \log \pi_\theta(a_{t}| o_{0:t-1})
\end{equation} is bounded. Thus, $ \nabla_\theta^2 \log \probs_\theta(y)$ is bounded by the triangle inequality.
Combining \eqref{eq:jensen} and \eqref{eq:bounded-hessian}, with similar reasoning for the boundedness of the gradient, it holds that $\nabla_\theta^2 H(S_0|Y; \theta)$
is bounded.
\end{proof}


  To obtain the locally optimal policy parameter $\theta$, we initialize a policy parameter $\theta_0$ and carry out the gradient descent.  At each iteration $\tau\ge 1$,  
  \begin{equation}
\theta_{\tau + 1} = \theta_{\tau} - \eta  \nabla_\theta H (S_0|Y;\theta_{\tau})
 \end{equation}
 where $\eta$ is the step size. 
When using the \emph{gradient descent} algorithm to compute the optimal $\theta$, it may be computationally expensive to compute $ \nabla_\theta H(S_0|Y;\theta)$ by enumerating all possible observations $y$. A sample approximation is employed such that at each iteration, we collect $M$ sequences of observations $\{y_1,\ldots, y_M\}$, and compute  
\begin{equation*}
\begin{aligned}
& \nabla_\theta H (S_0|Y;\theta) \approx \hat{\nabla}_\theta H (S_0|Y;\theta) \\  = & - \frac{1}{M} \sum_{k = 1}^M \sum_{s_0 \in \mathcal{S}_0} \Big[ \probs_\theta(s_0| y_k) \log \probs_\theta(s_0 | y_k) \nabla_\theta \log \probs_\theta(y_k) \Big].
\end{aligned}
\end{equation*}


\section{EXPERIMENTS} 
Consider a stochastic grid world environment (Fig.~\ref{fig:environment}) with three types of robots, each starting from different positions to reach specific goals (flags). The blue robot (type 1) starts at $(0,3)$, the red (type 2) at $(3,0)$, and the green (type 3) at $(5,2)$.
In line with standard grid world dynamics, each robot has a 20\% chance of moving to one of the two nearest cells instead of the intended direction. Robots hitting walls or boundaries remain in place. Their policies are computed to maximize the probability of reaching their goals from their starting positions.
\begin{figure}[htp!]
\centering
\includegraphics[width=0.45\textwidth]{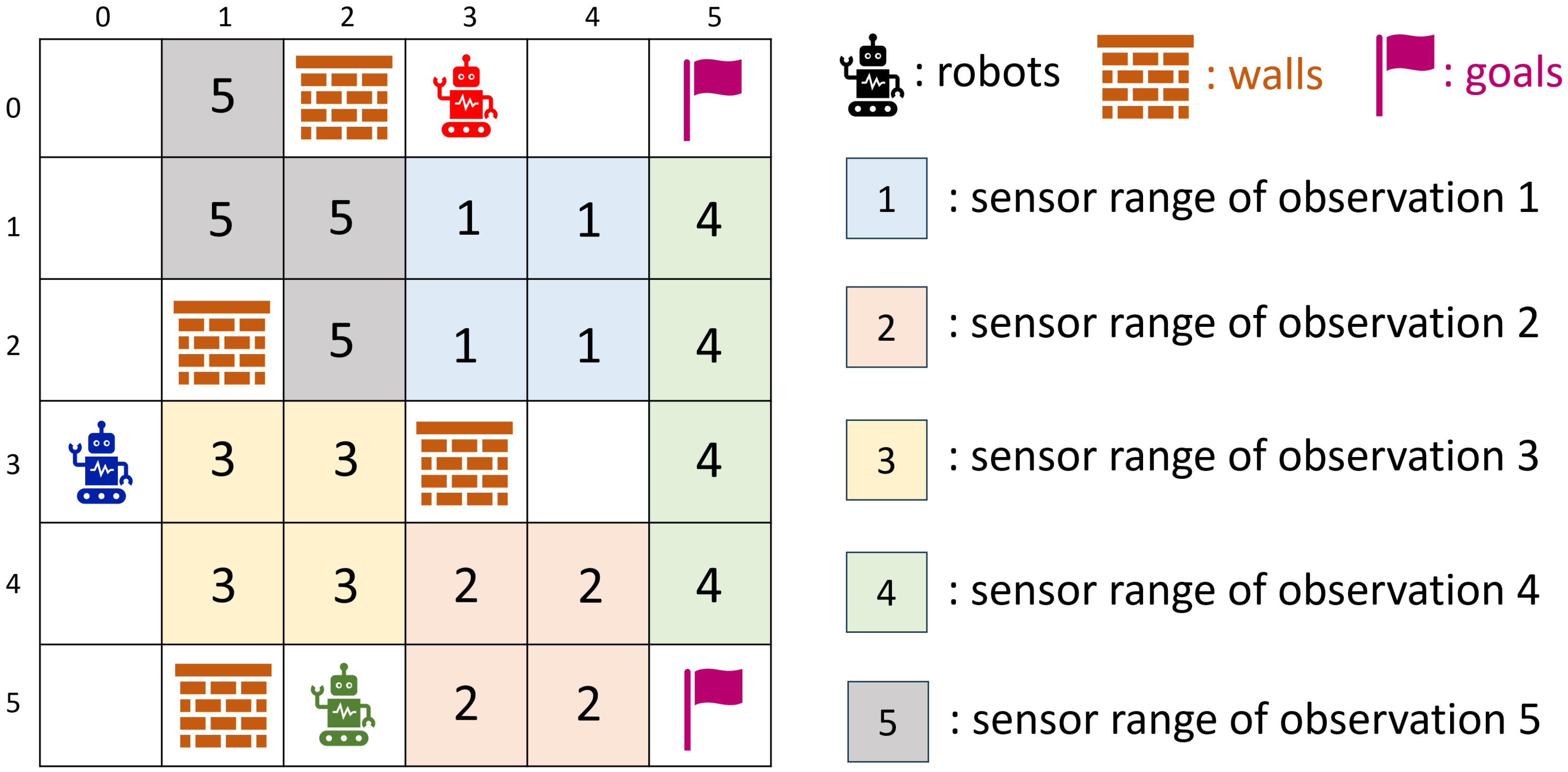}
\caption{A stochastic grid world monitored by a set of sensors.}
\label{fig:environment}
\end{figure}

\begin{figure*}[t!]
\centering
\begin{subfigure}{.24\textwidth}
  \centering
  \includegraphics[width=\linewidth]{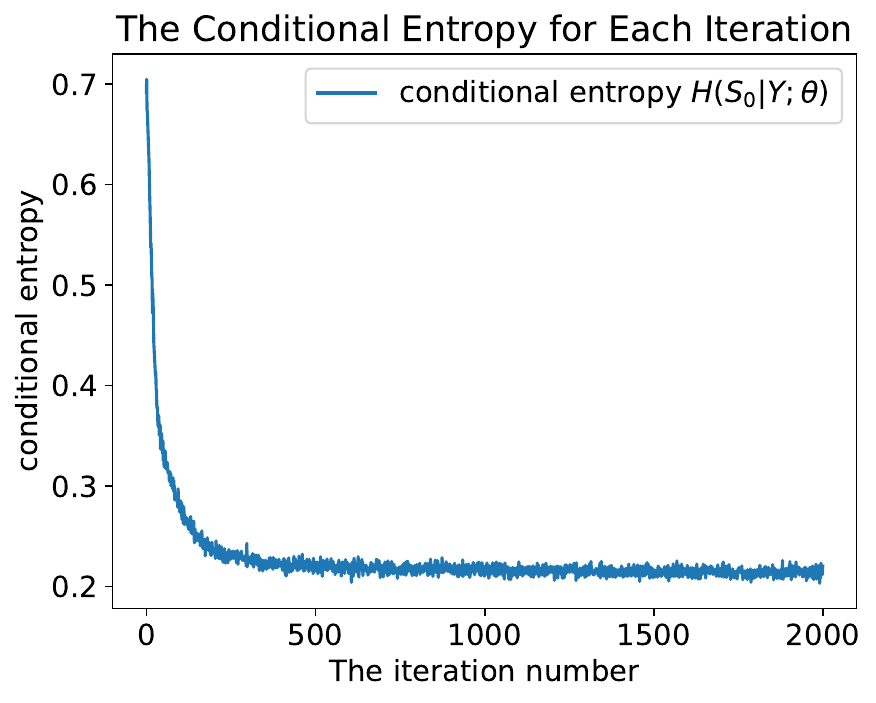}
  \caption{}
  \label{fig:opt_results}
\end{subfigure}%
\begin{subfigure}{.24\textwidth}
  \centering
  \includegraphics[width=\linewidth]{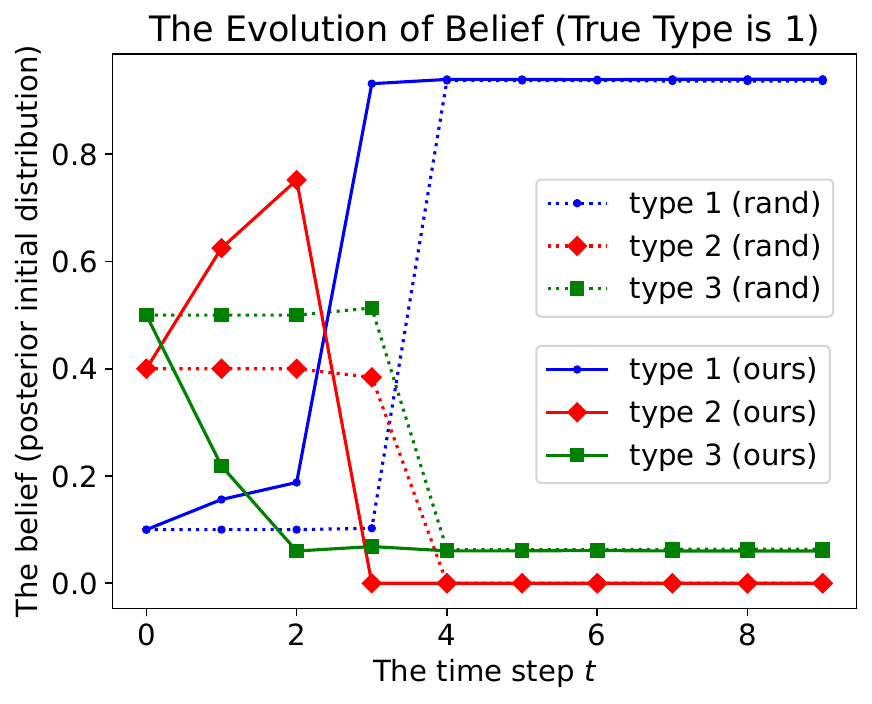}
  \caption{}
  \label{fig:type_1}
\end{subfigure}%
\begin{subfigure}{.24\textwidth}
  \centering
  \includegraphics[width=\linewidth]{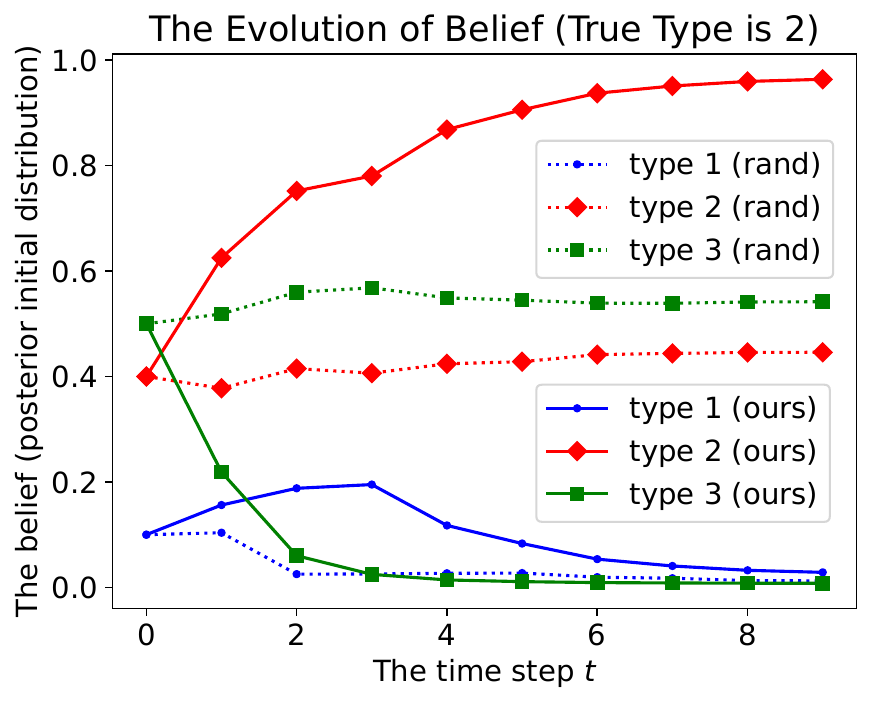}
  \caption{}
  \label{fig:type_2}
\end{subfigure}
\begin{subfigure}{.24\textwidth}
  \centering
  \includegraphics[width=\linewidth]{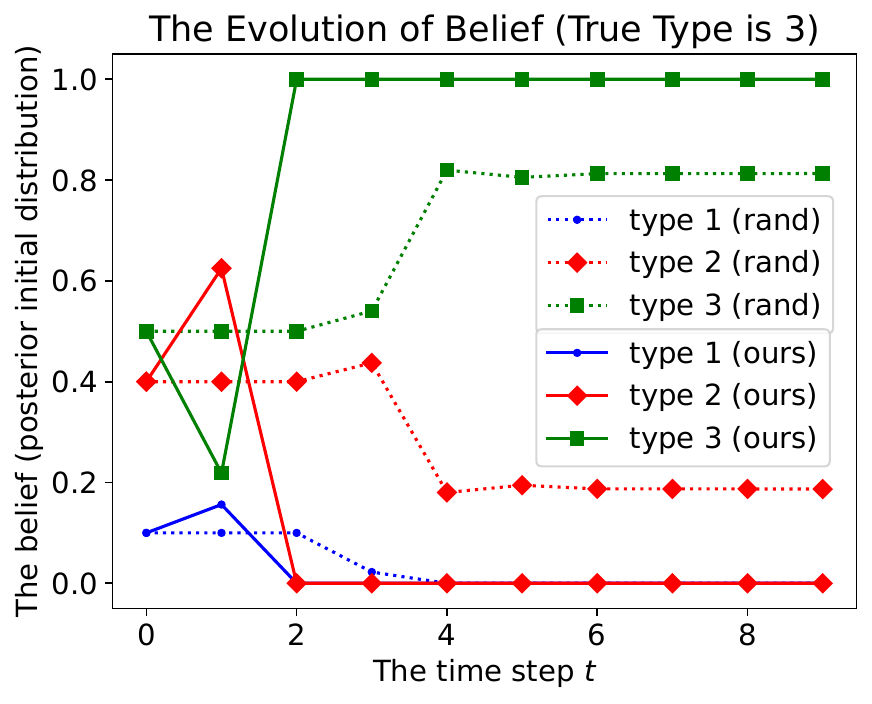}
  \caption{}
  \label{fig:type_3}
\end{subfigure}
\caption{The convergence results of policy gradients algorithm and beliefs evolution for different true types of robots.   }
\label{fig:baseline}
\vspace{-4ex}
\end{figure*}

The environment has five sensors, each with a distinct range. The observer can query only one sensor at a time. For sensor $i \in \{1,\ldots, 5\}$, if a robot is within its range and the observer queries it, the observer receives observation $i$ with 90\% probability and a null observation (``n") with 10\% probability (false negative). Otherwise, the observer  also receives a null observation.

We consider finite-state perception policy (Def.~\ref{def:perception_policy}): Given an integer $K\ge 0$, the state set $Q=\{O^{\le K}\}$ are the set of    observations with length $\le K$.  For each $q\in Q$, $\delta(q, o) = q'$ is defined such that $q' =  \mathsf{suffix}^{= K}(q\cdot o)$ \footnote{$\mathsf{suffix}^{= K}(w)$ is the last $K$ symbols of string $w$ if $|w|\ge K$ or $w$ itself otherwise.} is the last  (up to) $ K$  observations after appending the observation $o$ to  $q$. The   probabilistic output function is parameterized as, 
$      \psi_\theta(a|q) =  \frac{\exp(\theta_{q,a})}{\sum_{a' \in \mathcal{A}}\exp(\theta_{q, a'})},
$ where  $\theta \in \reals^{
|Q \times \mathcal{A}|}$ is the policy parameter vector. Given observations $o_{0:t}$, the policy $\pi_\theta(a| o_{0:t}) =  \psi_\theta( a| \delta(q_0, o_{0:t}))$.
The softmax policy satisfies Assumption~\ref{assume:bounded-policy-gradient} and is differentiable. 
In the experiments, we set the length of memory $K=2$.

The initial distribution $\mu_0$ is $0.1, 0.4, 0.5$ for type 1, type 2, and type 3, respectively. Figure~\ref{fig:opt_results} illustrates the convergence of the policy gradient \footnote{We sample $M = 2000$ trajectories and set the horizon $T = 10$ for each iteration. The fixed step size of the gradient descent algorithm is set to be $0.5$. We run $N = 2000$ iterations on the 12th Gen Intel(R) Core(TM) i7-12700, the average time consumed for one iteration is $6.7$ seconds.}.
When the algorithm converges, the conditional entropy $ H(S_0|Y;\theta) $ 
approaches approximately $0.22$. This indicates that the observations provide substantial information about the robot's type on average.



 
To provide an intuitive understanding, we use the computed perception policy $\theta^\ast$, called the ``min-entropy" policy, to evaluate the posterior belief $P_{\theta^\ast}(S_0|o_{0:i})$ of the robot's type given observations $o_{0:i}$ for $0 \le i \le T$. For comparison, we also generated random policies by randomly selecting policy parameter $\theta$ and chose the one with the lowest conditional entropy, referred to as the ``random policy." We then sampled an initial state/type and used both policies to collect observations. Figures \ref{fig:type_1}, \ref{fig:type_2}, and \ref{fig:type_3} show how the agent's belief evolves for different types under the random and min-entropy policies.
For each robot type, the min-entropy policy allows for accurate type identification by time $T=9$, with probabilities of $0.94$ for type 1, $0.96$ for type 2, and $1.00$ for type 3. In contrast, the random policy results in lower probabilities: $0.93$ for type 1, $0.45$ for type 2, and $0.81$ for type 3. Both policies perform well for type 1, but 
 the min-entropy policy significantly outperforms the random policy for identifying type 2 and type 3 robots.



\section{CONCLUSION AND FUTURE WORK}
In this paper, we introduce a conditional entropy measure to quantify uncertainty and formulate a problem of minimizing the uncertainty of initial states in an \ac{hmm}. To solve this optimization problem, we develop a gradient descent algorithm and derive the gradient of conditional entropy using observable operators. Further, we
prove that the conditional entropy is Lipschitz-continuous and $L$-smooth with respect to the policy parameters under certain assumptions for the policy search space.  
An interesting direction for future research would be to explore active perception under varying assumptions for the perception agent. For example, the agent with imprecise knowledge about the model dynamics. It is also interesting to consider a general POMDP formulation where the agent can change both the transition dynamics and emission function of the partially observable systems.

\bibliographystyle{IEEEtranS}
\bibliography{jiefu_refs, chongyang_refs}

\begin{thebibliography}{10}
\providecommand{\url}[1]{#1}
\csname url@samestyle\endcsname
\providecommand{\newblock}{\relax}
\providecommand{\bibinfo}[2]{#2}
\providecommand{\BIBentrySTDinterwordspacing}{\spaceskip=0pt\relax}
\providecommand{\BIBentryALTinterwordstretchfactor}{4}
\providecommand{\BIBentryALTinterwordspacing}{\spaceskip=\fontdimen2\font plus
\BIBentryALTinterwordstretchfactor\fontdimen3\font minus \fontdimen4\font\relax}
\providecommand{\BIBforeignlanguage}[2]{{%
\expandafter\ifx\csname l@#1\endcsname\relax
\typeout{** WARNING: IEEEtranS.bst: No hyphenation pattern has been}%
\typeout{** loaded for the language `#1'. Using the pattern for}%
\typeout{** the default language instead.}%
\else
\language=\csname l@#1\endcsname
\fi
#2}}
\providecommand{\BIBdecl}{\relax}
\BIBdecl

\bibitem{agarwal2021theory}
A.~Agarwal, S.~M. Kakade, J.~D. Lee, and G.~Mahajan, ``On the theory of policy gradient methods: Optimality, approximation, and distribution shift,'' \emph{Journal of Machine Learning Research}, vol.~22, no.~98, pp. 1--76, 2021.

\bibitem{andreopoulos2009theory}
A.~Andreopoulos and J.~K. Tsotsos, ``A theory of active object localization,'' in \emph{2009 IEEE 12th International Conference on Computer Vision}.\hskip 1em plus 0.5em minus 0.4em\relax IEEE, 2009, pp. 903--910.

\bibitem{Araya2010Neurips}
M.~Araya, O.~Buffet, V.~Thomas, and F.~Charpillet, ``A {POMDP} extension with belief-dependent rewards,'' in \emph{Advances in Neural Information Processing Systems}, J.~Lafferty, C.~Williams, J.~Shawe-Taylor, R.~Zemel, and A.~Culotta, Eds., vol.~23.\hskip 1em plus 0.5em minus 0.4em\relax Curran Associates, Inc., 2010.

\bibitem{bajcsyRevisitingActivePerception2018}
R.~Bajcsy, Y.~Aloimonos, and J.~K. Tsotsos, ``\BIBforeignlanguage{en}{Revisiting active perception},'' \emph{\BIBforeignlanguage{en}{Autonomous Robots}}, vol.~42, no.~2, pp. 177--196, Feb. 2018.

\bibitem{CASAO2024103876}
S.~Casao, Álvaro Serra-Gómez, A.~C. Murillo, W.~Böhmer, J.~Alonso-Mora, and E.~Montijano, ``Distributed multi-target tracking and active perception with mobile camera networks,'' \emph{Computer Vision and Image Understanding}, vol. 238, p. 103876, 2024.

\bibitem{6161683}
M.~Dunbabin and L.~Marques, ``Robots for environmental monitoring: Significant advancements and applications,'' \emph{IEEE Robotics \& Automation Magazine}, vol.~19, no.~1, pp. 24--39, 2012.

\bibitem{egorovTargetSurveillanceAdversarial2016}
M.~Egorov, M.~J. Kochenderfer, and J.~J. Uudmae, ``Target surveillance in adversarial environments using {POMDPs},'' in \emph{Proceedings of the {Thirtieth} {AAAI} {Conference} on {Artificial} {Intelligence}}.\hskip 1em plus 0.5em minus 0.4em\relax AAAI Press, 2016, pp. 2473--2479.

\bibitem{jaeger2000observableoperator}
H.~Jaeger, ``{Observable Operator Models for Discrete Stochastic Time Series},'' \emph{Neural Computation}, vol.~12, no.~6, pp. 1371--1398, 06 2000.

\bibitem{jaegerObservableOperatorProcesses1997}
H.~Jaeger and S.~Augustin, ``\BIBforeignlanguage{en}{Observable {Operator} {Processes} and {Conditioned} {Continuation} {Representations}},'' \emph{\BIBforeignlanguage{en}{Arbeitspapiere der GMD}}, no. 1043, Jan. 1997.

\bibitem{lafortuneHistoryDiagnosabilityOpacity2018a}
\BIBentryALTinterwordspacing
S.~Lafortune, F.~Lin, and C.~N. Hadjicostis, ``On the history of diagnosability and opacity in discrete event systems,'' \emph{Annual Reviews in Control}, vol.~45, pp. 257--266, Jan. 2018. [Online]. Available: \url{https://www.sciencedirect.com/science/article/pii/S136757881830004X}
\BIBentrySTDinterwordspacing

\bibitem{lauriMultiagentActivePerception2020}
M.~Lauri and F.~Oliehoek, ``Multi-agent active perception with prediction rewards,'' in \emph{Advances in {Neural} {Information} {Processing} {Systems}}, vol.~33.\hskip 1em plus 0.5em minus 0.4em\relax Curran Associates, Inc., 2020, pp. 13\,651--13\,661.

\bibitem{molloy2023smoother}
T.~Molloy and G.~Nair, ``Smoother entropy for active state trajectory estimation and obfuscation in pomdps,'' \emph{IEEE Transactions on Automatic Control}, vol.~PP, pp. 1--16, 06 2023.

\bibitem{SavasEntropy2022}
Y.~Savas, M.~Hibbard, B.~Wu, T.~Tanaka, and U.~Topcu, ``Entropy maximization for partially observable markov decision processes,'' \emph{IEEE Transactions on Automatic Control}, vol.~67, no.~12, pp. 6948--6955, 2022.

\bibitem{Shendeep2019}
M.~Shen and J.~P. How, ``Active perception in adversarial scenarios using maximum entropy deep reinforcement learning,'' in \emph{2019 International Conference on Robotics and Automation (ICRA)}.\hskip 1em plus 0.5em minus 0.4em\relax IEEE Press, 2019, p. 3384–3390.

\bibitem{shi2024information}
C.~Shi, Y.~Bu, and J.~Fu, ``Information-theoretic opacity-enforcement in markov decision processes,'' in \emph{Proceedings of the Thirty-Third International Joint Conference on Artificial Intelligence, {IJCAI-24}}, K.~Larson, Ed.\hskip 1em plus 0.5em minus 0.4em\relax International Joint Conferences on Artificial Intelligence Organization, 8 2024, pp. 6779--6787, main Track.

\bibitem{zhou2011multirobot}
K.~Zhou and S.~I. Roumeliotis, ``Multirobot active target tracking with combinations of relative observations,'' \emph{IEEE Transactions on Robotics}, vol.~27, no.~4, pp. 678--695, 2011.

\end{thebibliography}

\end{document}